\documentclass[conference]{IEEEtran}

\usepackage{graphicx}
\usepackage{graphics}
\usepackage{mdwlist}
\usepackage{algorithm}
\usepackage[noend]{algpseudocode}
\usepackage{booktabs}
\usepackage[center]{subfigure}

\usepackage{amsthm,amssymb}
\usepackage[]{caption}
\usepackage{multirow}
\usepackage{mathrsfs}
\usepackage{amsbsy}
\usepackage[mathscr]{eucal} %% For \mathscr
\usepackage{xspace}
\usepackage[dvipsnames]{xcolor} %check mark
\usepackage{paralist} %compact item
\usepackage{mathtools}

\usepackage[hidelinks,bookmarks=false]{hyperref}
\usepackage{pifont}% http://ctan.org/pkg/pifont
\usepackage{balance}

\renewcommand\IEEEkeywordsname{Keywords}

\newtheorem{problem}{Problem}
\newtheorem{definition}{Definition}
\newtheorem{lemma}{Lemma}
\newtheorem{theorem}{Theorem}

\newcommand{\hide}[1]{}

\newcommand{\cintlong}{Closing Interval\xspace}
\newcommand{\fintlong}{Total Interval\xspace}

\newcommand{\SG}{\mathcal{G}}
\newcommand{\SV}{\mathcal{V}}
\newcommand{\SE}{\mathcal{E}}
\newcommand{\SSS}{\mathcal{S}}
\newcommand{\SW}{\mathcal{W}}
\newcommand{\SR}{\mathcal{R}}
\newcommand{\SGT}{\mathcal{G}^{(t)}}
\newcommand{\SVT}{\mathcal{V}^{(t)}}
\newcommand{\ET}{e^{(t)}}
\newcommand{\PT}{p^{(t)}}
\newcommand{\SET}{\mathcal{E}^{(t)}}

\newcommand{\triple}{(u,v,w)}
\newcommand{\type}{type_{uvw}}
\newcommand{\tprob}{p_{uvw}}
\newcommand{\tcount}{x_{uvw}}

\newcommand{\tgen}{t^{(i)}_{uvw}}
\newcommand{\tmax}{t^{(3)}_{uvw}}
\newcommand{\tmed}{t^{(2)}_{uvw}}
\newcommand{\tmin}{t^{(1)}_{uvw}}

\newcommand{\STT}{\mathcal{T}^{(t)}}
\newcommand{\CT}{c^{(t)}}
\newcommand{\SN}{\hat{\mathcal{N}}}
\newcommand{\SGH}{\hat{\mathcal{G}}}
\newcommand{\SVH}{\hat{\mathcal{V}}}
\newcommand{\SEH}{\hat{\mathcal{E}}}
\newcommand{\bit}{\begin{itemize*}}
\newcommand{\eit}{\end{itemize*}}
\newcommand{\ben}{\begin{enumerate*}}
\newcommand{\een}{\end{enumerate*}}

\newcommand{\method}{\textsc{WRS}\xspace}

\newcommand{\mascot}{\textsc{MASCOT}\xspace}
\newcommand{\triest}{\textsc{Triest-IMPR}\xspace}

\newcommand{\bluecolor}{\textcolor{blue}}

\begin{document}
%\setcopyright{acmcopyright}

\title{WRS: Waiting Room Sampling for Accurate Triangle Counting in Real Graph Streams}

%\title{Bare Demo of IEEEtran.cls\\ for IEEE Conferences}

% author names and affiliations
% use a multiple column layout for up to three different
% affiliations

\author{\IEEEauthorblockN{Kijung Shin}
\IEEEauthorblockA{Carnegie Mellon University, Pittsburgh, PA, USA \\
Email: kijungs@cs.cmu.edu}}

% ACM Format
%\numberofauthors{2}
%\author{Kijung Shin}
%\affiliation{%
%	\institution{School of Computer Science, Carnegie Mellon University}
%	\city{Pittsburgh} 
%	\state{PA}
%	\country{USA}
%}
%\email{kijungs@cs.cmu.edu}

\maketitle
\begin{abstract}
%Given a graph stream, how can we estimate the global and local counts of triangles with fixed amount of memory?
%Especially, for accurate estimation, which edges should be prioritized when we cannot store all edges in memory?
If we cannot store all edges in a graph stream, which edges should we store to estimate the triangle count accurately?
%Can we do better than uniform random sampling?

Counting triangles (i.e., cycles of length three) is a fundamental graph problem with many applications in social network analysis, web mining, anomaly detection, etc.
Recently, much effort has been made to accurately estimate global and local triangle counts in streaming settings with limited space. %, especially, in a large-scale graph stream not fitting in memory.
Although existing methods use sampling techniques without considering temporal dependencies in edges,
%are based on uniform random sampling ignoring, 
we observe {\em temporal locality} in real dynamic graphs.
That is, 
future edges are more likely to form triangles 
with recent edges than with older edges.
% in dynamic graph streams, where new edges are streamed in the order they are created.

In this work, we propose a single-pass streaming algorithm called {\em Waiting-Room Sampling} (\method) for global and local triangle counting.
%\method exploits the temporal locality to reduce the variances of the unbiased estimates it provides.
%Specifically,
\method exploits the temporal locality by always storing the most recent edges, which future edges are more likely to form triangles with, in the {\em waiting room}, while it uses reservoir sampling for the remaining edges.
%It divides available memory into two parts so that the most recent edges are stored in the first part (called ``waiting room"), while the remaining edges
%and the remaining edges are stored, respectively, with different probabilities.
Our theoretical and empirical analyses show that \method is: {\em (a) Fast and `any time':} runs in linear time, always maintaining and updating estimates while new edges arrive, {\em (b) Effective}: yields up to {\em 47\% smaller estimation error} than its best competitors, and {\em (c) Theoretically sound}: gives unbiased estimates with small variances under the temporal locality.
%successfully reduced variances, while preserving unbiasedness, and thus produced up to $50\%$ smaller errors than its best competitors in our experiments.
% discovery maybe?

\end{abstract}


\section{Introduction}
\label{sec:intro}
\begin{figure*}[t]
	\centering
	\vspace{-4mm}
	\subfigure[\method is fast and `any time']{\label{fig:crown:speed}
		\includegraphics[width=0.1825\linewidth]{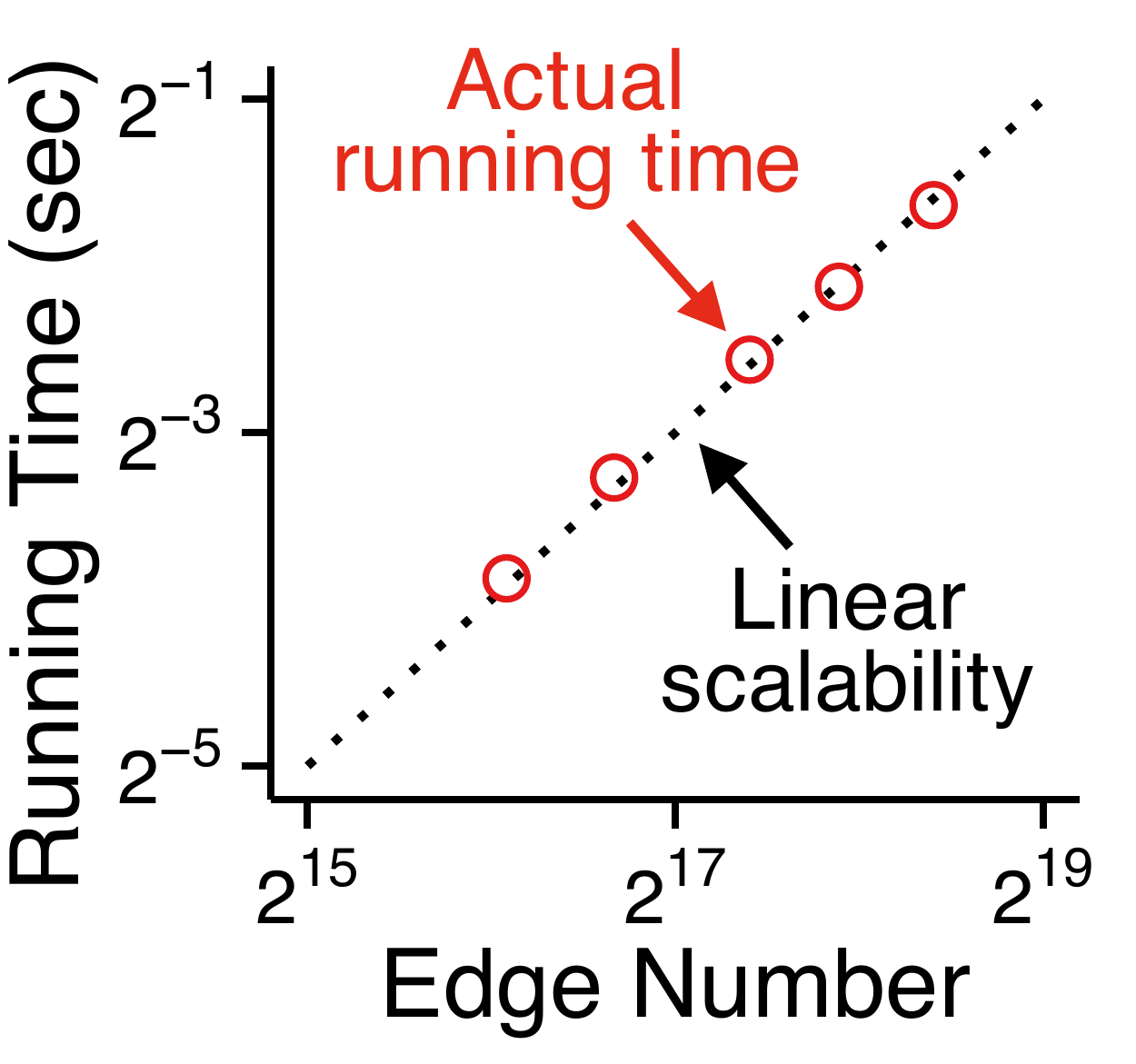}
		\includegraphics[width=0.1825\linewidth]{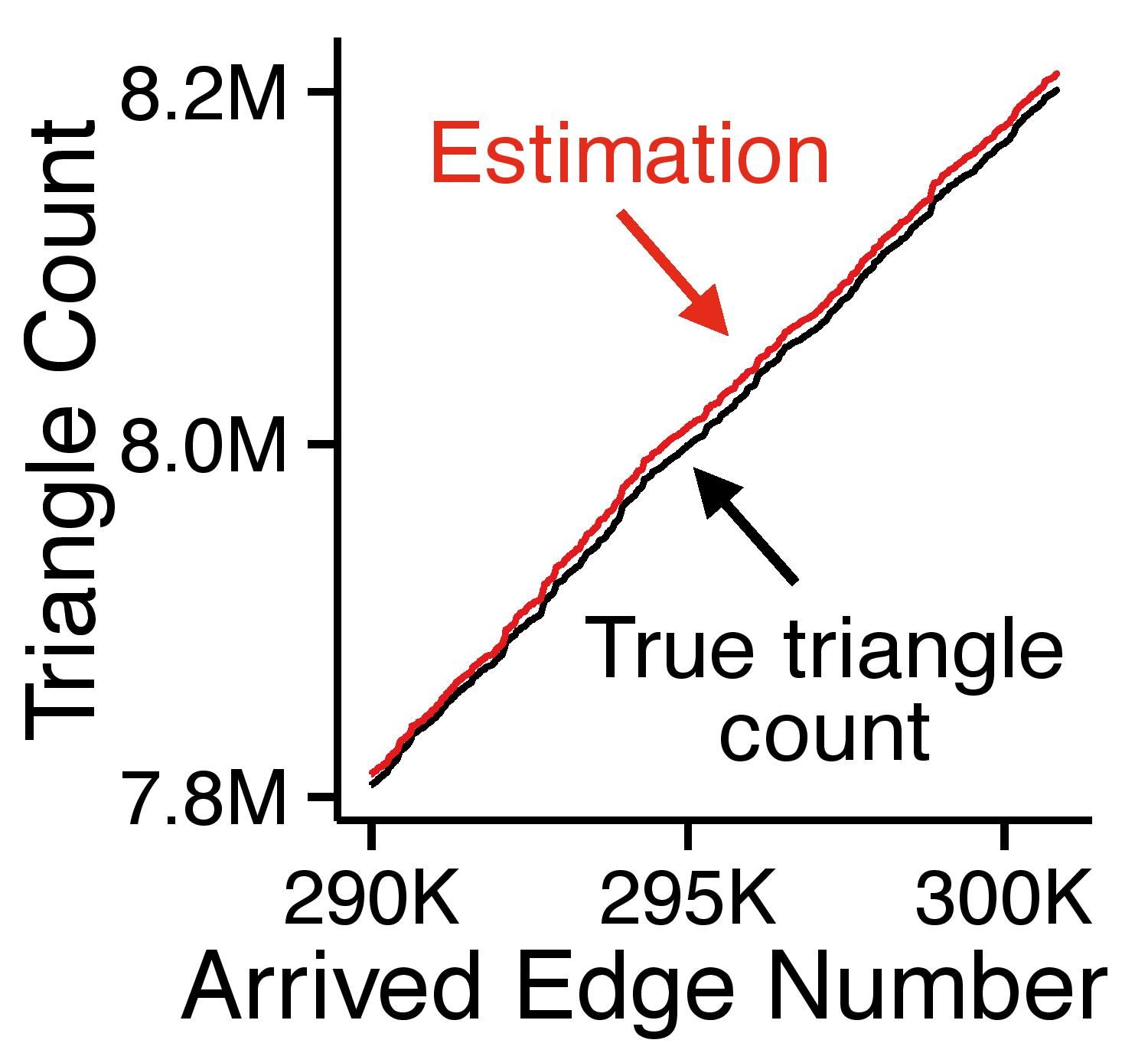}
	}
	\subfigure[\method is effective]{\label{fig:crown:accuracy}
		\includegraphics[width=0.1825\linewidth]{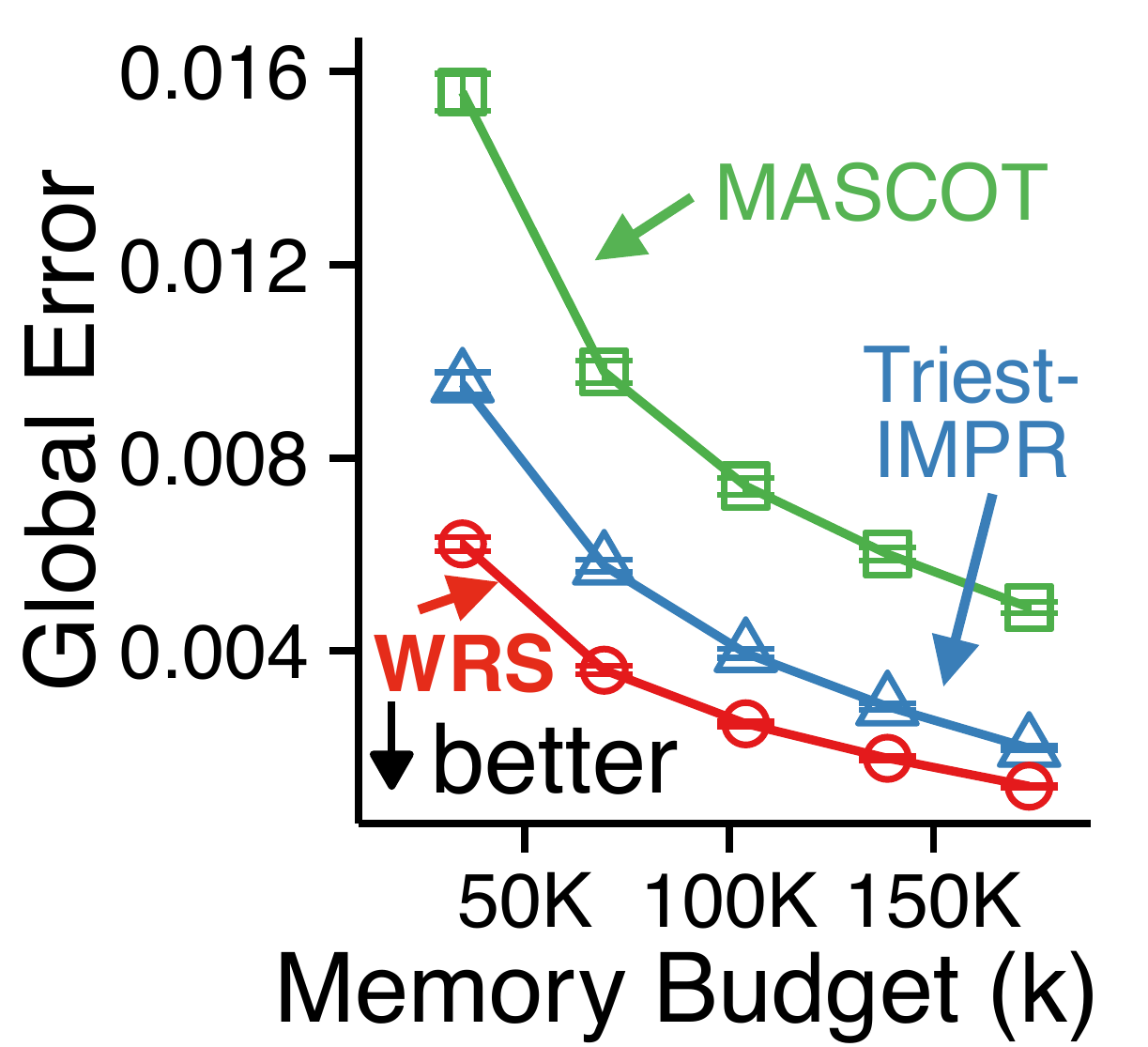}
		\includegraphics[width=0.1825\linewidth]{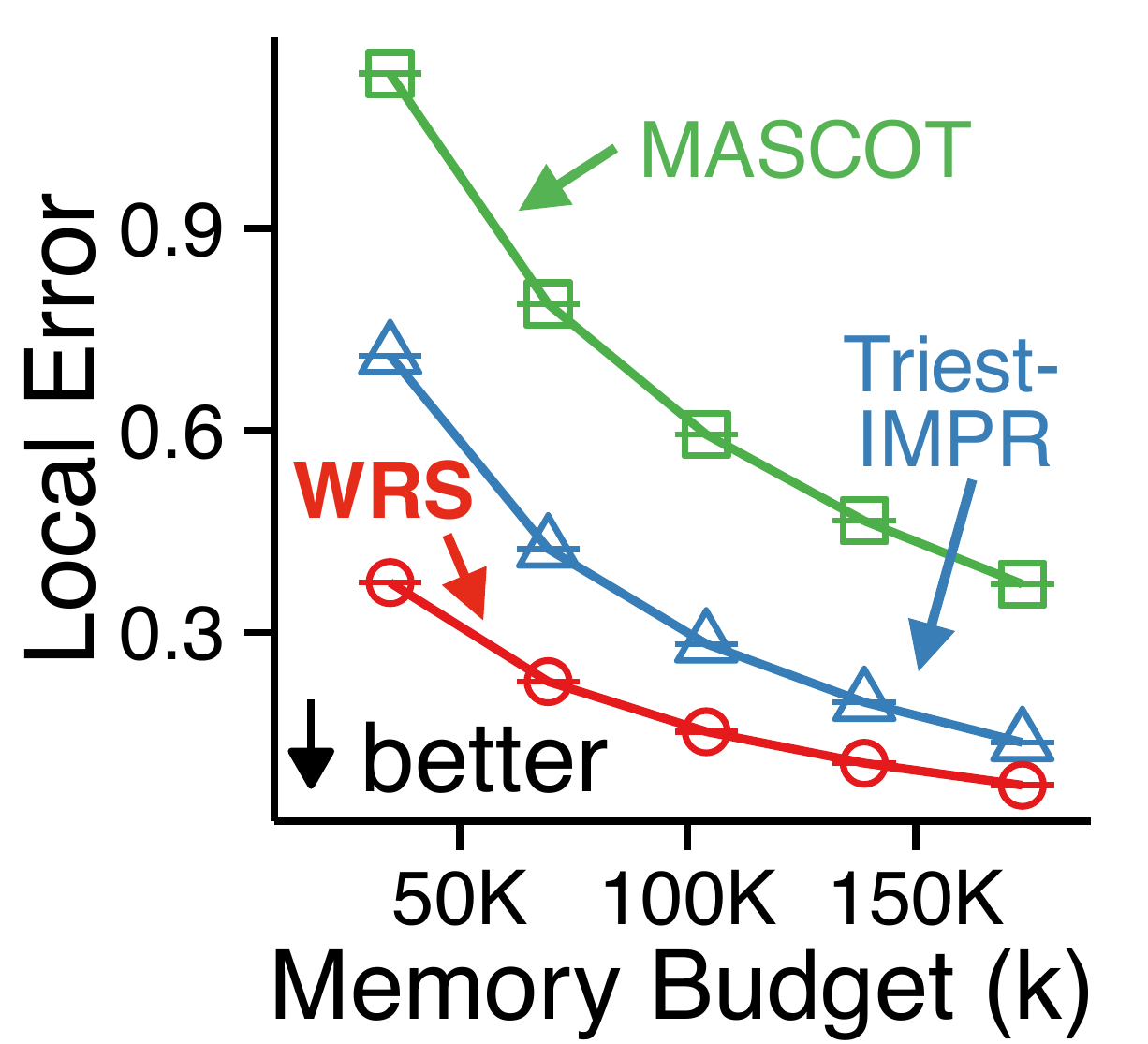}
	}
	\subfigure[\method is unbiased]{\label{fig:crown:analysis}
		\includegraphics[width=0.1825\linewidth]{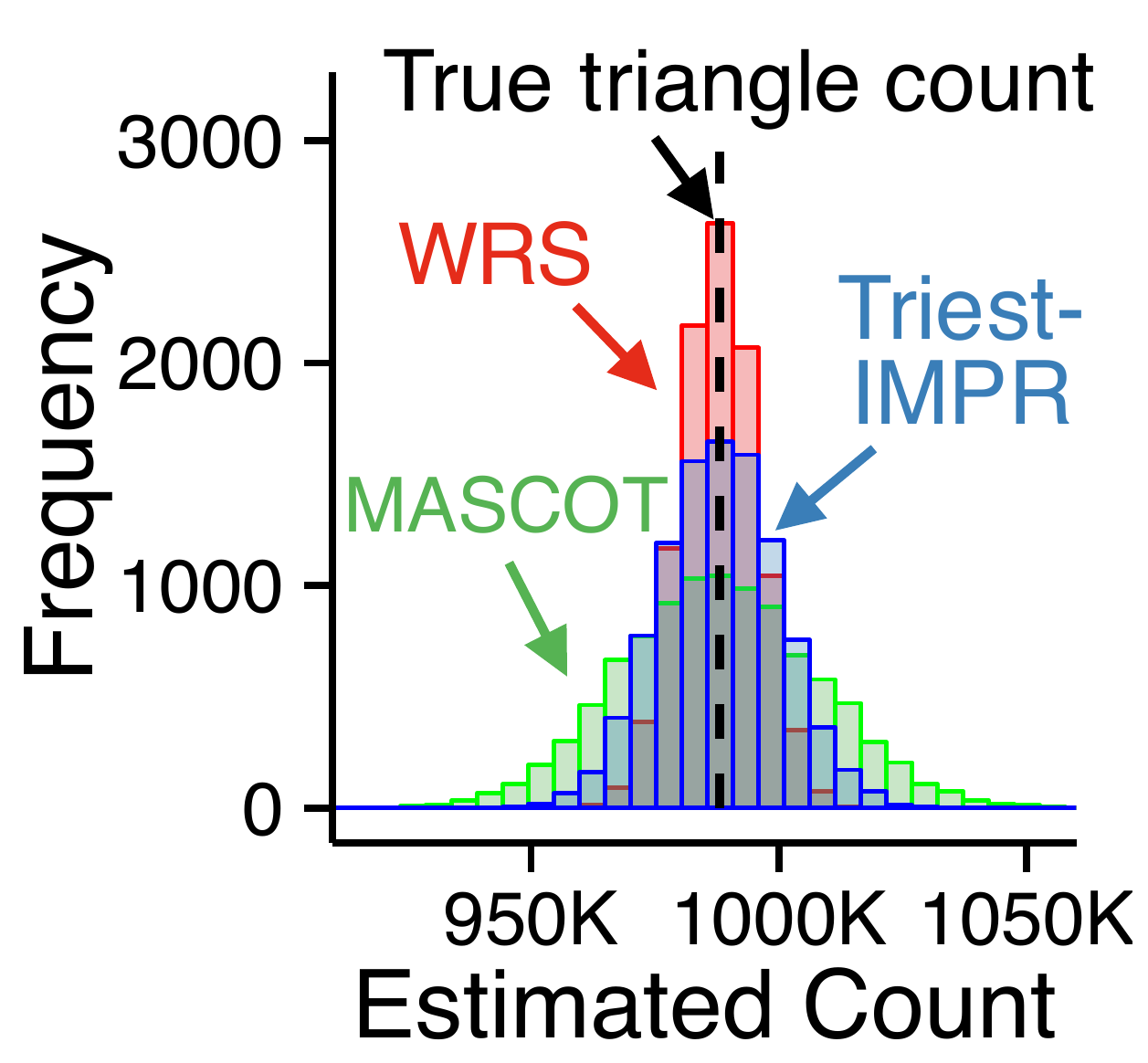}
	}\\
	\vspace{-2mm}
	\caption{\label{fig:crown}
		\underline{\smash{Strengths of \method.}} (a) \method scales linearly with the number of edges, and it always maintains the estimates of the triangle counts while input graphs grow with new edges. (b) \method is more accurate than state-of-the-art streaming algorithms in global and local triangle counting. (c) \method gives unbiased estimates (Theorem~\ref{thm:unbias}) with small variances (Lemma~\ref{lemma:variance:type}). The ArXiv dataset, explained in the experiment section, is used for all these figures.
	}
\end{figure*}

In a dynamic graph stream, where new edges are streamed as they are created, how can we count the triangles?
If we cannot store all the edges in memory, which edges should we store to estimate the triangle count accurately?
%Is there a better way than uniform random sampling?

Counting the triangles (i.e., cycles of length $3$) in a graph is a fundamental problem with many applications.
For example, triangles in social networks have received much attention as an evidence of homophily (i.e., people choose friends similar to themselves) \cite{mcpherson2001birds} and transitivity (i.e., people with common friends become friends) \cite{wasserman1994social}.
Thus, many concepts in social network analysis, such as social balance \cite{wasserman1994social}, the clustering coefficient \cite{watts1998collective}, and the transitivity ratio \cite{newman2003structure}, are based on the count of triangles.
Moreover, the count of triangles has been used for spam detection \cite{becchetti2010efficient}, web-structure analysis \cite{eckmann2002curvature}, degeneracy estimation \cite{shin2016corescope}, and query optimization \cite{bar2002reductions}. 

Due to the importance of triangle counting, numerous algorithms have been developed in many different settings, including multi-core \cite{shun2015multicore,kim2014opt}, %rahman2013approximate
external-memory \cite{kim2014opt,hu2014efficient}, distributed-memory \cite{arifuzzaman2013patric}, %gonzalez2014graphx
and MapReduce \cite{suri2011counting,park2014mapreduce,park2016pte} %cohen2009graph
settings.
The algorithms aim to accurately and rapidly count {\em global triangles} (i.e., all triangles in a graph) and/or {\em local triangles} (i.e., triangles that each node is involved with).%, which are directly related to local clustering coefficients \cite{kutzkov2013streaming}.

Especially, as many real graphs, including social media and web, evolve over time, recent work has focused largely on streaming settings where graphs are given as streams of new edges.
To accurately estimate the count of the triangles in large graph streams not fitting in memory, various sampling techniques have been developed \cite{kutzkov2013streaming,lim2015mascot,de2016tri}.

However, existing streaming algorithms sample edges without considering temporal dependencies in edges and thus cannot exploit {\em temporal locality}, i.e., the tendency that future edges are more likely to form triangles with recent edges than with older edges.
This temporal locality is observed commonly in many real dynamic graph streams, where new edges are streamed as they are created.

Then, how can we exploit the temporal locality for accurately estimating global and local triangle counts?
We propose {\em Waiting-Room Sampling} (\method), a single-pass streaming algorithm that always stores the most recent edges in the {\em waiting room}, while it uses standard reservoir sampling \cite{vitter1985random} for the remaining edges.
The waiting room increases the probability that, when a new edge arrives, edges forming triangles with the new edge are in memory.
Reservoir sampling, on the other hand, enables \method to yield unbiased estimates.

%By doing so, \method yields unbiased estimates with smaller variances than its best competitors, as shown in our theoretical and experimental analyses.

Specifically, our theoretical and empirical analyses show that \method has the following strengths:
\bit
	\item {\bf Fast and `any time'}:  \method runs in linear time in the number of edges, giving the estimates of triangle counts at any time, not only at the end of streams (Figure~\ref{fig:crown:speed}).
	\item {\bf Effective}: \method produces up to {\em 47\% smaller} estimation error than its best competitors (Figure~\ref{fig:crown:accuracy}).
	\item {\bf Theoretically sound}: we prove the unbiasedness of the estimators provided by \method and their small variances under the temporal locality (Theorem~\ref{thm:unbias} and Figure~\ref{fig:crown:analysis}).
\eit
{\bf Reproducibility}: The code and datasets used in the paper are available at \textit{\url{http://www.cs.cmu.edu/~kijungs/codes/wrs/}}.

%In summary, the contributions of our study are as follows:

%Specifically, our theoretical and experimental analysis show that \method is: (a) {\bf ``any time'':} giving the unbiased estimates of global and local triangle counts at any time, not only at the end of the stream,
%(b) {\bf space efficient:} using constant amount of space, which is fully utilized by reservoir sampling, (c) {\bf fast:}  per-edge processing time of \method scales linearly with the sample count not depending on the size of streams, and (d) {\bf accurate:} producing up to $50\%$ smaller errors than its best competitors (see Figure~\ref{fig:crown:accuracy}).
%\bit
%	\item {\bf Observation}:  we discover temporal locality in triangle formulation in dynamic graph streams (Figure~\ref{fig:crown:dist}).
%	\item {\bf Algorithm}: we propose \method, a single-pass streaming algorithm utilizing the temporal locality for accurate global and local triangle counting.
%	\item {\bf Analysis}: we prove the unbiasedness of \method (Figure~\ref{fig:crown:analysis}) and its variance in terms of input factors.
%\eit
%{\bf Reproducibility}: The code and data
%we used in the paper are available at \url{http://www.cs.cmu.edu/~kijungs/triangle}.

In Section~\ref{sec:related}, we review related work.
In Section~\ref{sec:prelim}, we introduce notations and the problem definition. 
In Section~\ref{sec:pattern}, we discuss temporal locality in real graph streams.
In Section~\ref{sec:method}, we propose our algorithm \method.
After showing experimental results in Section~\ref{sec:experiments}, 
we draw conclusions in Section~\ref{sec:conclusion}.

\section{Related Work}
\label{sec:related}
In this section, we discuss previous work on global and local triangle counting in a graph stream with limited space.
%, which is directly related to our work, although triangle counting has been researched in many different settings, including multi-core \cite{rahman2013approximate,shun2015multicore,kim2014opt}, external-memory \cite{kim2014opt,hu2014efficient}, distributed-memory \cite{arifuzzaman2013patric,gonzalez2014graphx}, and MapReduce \cite{cohen2009graph,suri2011counting,park2013efficient,park2014mapreduce,park2016pte} settings.
%We refer to \cite{mcgregor2014graph,zhang2010survey} for discussion
%of general graph streaming algorithms.

{\bf Global triangle counting in graph streams:}
For estimating the count of global triangles (i.e., all triangles in a graph), Tsourakakis et al. \cite{tsourakakis2009doulion} proposed sampling each edge i.i.d. with probability $p$. Then, the global triangle count can be estimated simply by multiplying that in the sampled graph by $p^{-3}$.
Jha et al. \cite{jha2013space} and Pavan et al. \cite{pavan2013counting} proposed sampling wedges (i.e., paths of length $2$) instead of edges for better space efficiency.
Kutzkov and Pagh \cite{kutzkov2014triangle} combined edge and wedge sampling methods for global triangle counting in graph streams where edges can be both inserted or deleted.

%{\bf Algorithms for Global Triangle Counting:}
%For counting global triangles (i.e., all triangles in a graph), Tsourakakis et al. \cite{tsourakakis2009doulion} used a graph-sparsification method based on edge sampling.
%Although the algorithm does not assume a streaming setting, it can be applied to a streaming setting since it requires a single pass of a graph.
%Jha et al. \cite{jha2013space} proposed an algorithm based on wedge sampling for global triangle counting in graph streams. It uses up to $O(\sqrt{|\SV|})$ space with an accuracy guarantee.
%A similar approach was proposed by Pavan et al. \cite{pavan2013counting} for counting constant-sized cliques, including triangles.
%Kutzkov and Pagh \cite{kutzkov2014triangle} combined triplet sampling and sparsification for global triangle counting in fully dynamic graph streams, where edges can be both inserted or deleted. The algorithm, however, uses more space than what is needed to store the entire graph in the worst case, as pointed in \cite{de2016tri}.
%Ahmed et al. \cite{ahmed2014graph} proposed a generic stream sampling framework for estimating many graph properties including the count of global triangles.
%Recently, Etemadi et al. \cite{etemadi2016efficient} proposed a non-uniform edge-sampling method, which however, is not applicable to streaming settings.

{\bf Local triangle counting in graph streams:}
For estimating the count of local triangles (i.e., triangles with each node),  Lim and Kang \cite{lim2015mascot} proposed sampling each edge i.i.d with probability $p$ but updating global and local counts whenever an edge arrives, even when the edge is not sampled. To properly set $p$, however, the number of edges in input streams should be known in advance. 
Likewise, randomly coloring nodes to sample the edges connecting nodes of the same color, as suggested by Kutzkov and Pagh \cite{kutzkov2013streaming}, requires the number of nodes in advance to decide the number of colors.
De Stefani et al. \cite{de2016tri} solved this problem using reservoir sampling \cite{vitter1985random}, which fully utilizes memory space within a budget, without requiring any prior knowledge.
In addition to these single-pass streaming algorithms, semi-streaming algorithms with multiple passes over a graph were also explored \cite{becchetti2010efficient,kolountzakis2010efficient}.

%{\bf Algorithms for Local Triangle Counting:}
%For counting local triangles (i.e., triangles that each node is involved with),  Kolountzakis et al. \cite{kolountzakis2010efficient} and Becchetti et al. \cite{becchetti2010efficient} proposed semi-streaming algorithms, which require more than one passes of a graph. 
%A single-pass streaming algorithm was developed by Kutzkov and Pagh \cite{kutzkov2013streaming}.
%The algorithm randomly colors nodes and samples each edge if its endpoints have the same color. This algorithm, however, requires to know the number of
%nodes in advance, as pointed in \cite{lim2015mascot}.
%Lim and Kang \cite{lim2015mascot} proposed an single-pass algorithm called \mascot, which uses uniform random sampling of edges but updates global and local counts whenever an edge arrives even when the edge is not sampled.
%\mascot requires a fixed sampling probability $p$, but to properly set $p$, it is required to know the number of edges in a stream in advance.
%De Stefani et al. \cite{de2016tri} solved this problem using reservoir sampling \cite{vitter1985random}, which fully utilizes memory space within a budget, regardless of the length of a graph stream.
%Specifically, they proposed \triestfd, which supports edge deletions, and \triest, which does not support edge deletions but achieves higher accuracy.

Our single-pass algorithm \method estimates both global and local triangle counts in a graph stream without any prior knowledge about the input graph stream (see Section~\ref{sec:prelim:problem} for the detailed settings). Different from the existing approaches above, \method exploits the temporal locality in real graph streams (see Section~\ref{sec:pattern}), leading to higher accuracy.
%This makes \method more accurate than the existing approaches, as we show in the following sections. 

%To the best of our knowledge, our metric is the rst
%one that is based on the distribution of dense subgraphs in realworld
%graphs.

%Our single-pass algorithm \method does not require any knowledge about the input graph stream (node count, edge count, etc).
%Moreover, it gives estimates of the global and local triangle counts at any time before reaching the end of the stream and thus can be used for infinite graph streams. 
%Most importantly, \method gives the best accuracy by exploiting `locality' in real-world graph streams (see Section~\ref{sec:pattern}), which has been overlooked in existing work.

\section{Preliminaries}
\label{sec:prelim}
In this section, we first introduce notations and concepts used in the paper.
Then, we define the problem of global and local triangle counting in a real dynamic graph stream.
%Finally, we discuss previous approaches for the problems.

\begin{table}[!t]
	\centering
	\caption{Table of symbols.}
	\begin{tabular}{l|l}
		\toprule
		\textbf{Symbol} & \textbf{Definition} \\
		%\midrule
		%\multicolumn{2}{l}{Notations for a Static Tensor} \\
		\midrule
		\multicolumn{2}{l}{Notations for Graph Streams} \\
		\midrule
		$\SGT=(\SVT,\SET)$ & graph $\SG$ at time $t$\\
		$\ET$ & edge arriving at time $t$\\
		$(u,v)$ & edge between nodes $u$ and $v$ \\
		$t_{uv}$ & arrival time of edge $(u,v)$\\
		$\triple$ & triangle with nodes $u$, $v$, and $w$ \\
		$\tgen$ & arrival time of the $i$-th edge in $\triple$ \\
		%$t_{uvw}$ & arrival time of triangle $\triple$\\
		$\STT$ & set of triangles in $\SG^{(t)}$\\
		$\STT_{u}$ & set of triangles with node $u$ in $\SG^{(t)}$\\
		\midrule
		\multicolumn{2}{l}{Notations for Our Algorithm (defined in Section~\ref{sec:method})} \\
		\midrule
		$\SSS$ & given memory space \\
		$\SW$ & waiting room \\
		$\SR$ & reservoir \\
		$k$ & maximum number of edges stored in $\SSS$ \\
		$\alpha$ & relative size of the waiting room (i.e., ${|\SW|}/{|\SSS|}$) \\
		$\SGH=(\SVH,\SEH)$ & graph composed of the edges in $\SSS$ \\
		$\SN_{u}$ & set of neighbors of node $u$ in $\SGH$\\
%		$\SN_{uv}$ & $\SN_{u}\cap \SN_{v}$ \\
%		\midrule
%		\multicolumn{2}{l}{Notations for Analysis (defined in Section~\ref{sec:method})}  \\
%		\midrule
%		$\type$ & type of triangle $\triple$\\
%		$p^{(t)}_{uv}$ & sampling prob. of $(u,v)$ at time $t$\\
%		$\tcount$ & counts increased by triangle $\triple$\\
%		$\CT$ & estimate of $|\STT|$ \\
%		$\CT_{u}$ & estimate of $|\STT_{u}|$ \\
		\bottomrule
	\end{tabular}
	\label{tab:symbols}
\end{table}

\subsection{Notations and Concepts}
\label{sec:prelim:notation}

Symbols frequently used in the paper are listed in Table~\ref{tab:symbols}.
Consider an undirected graph $\SG=(\SV,\SE)$ with the set of nodes $\SV$ and the set of edges $\SE$.
We use unordered pair $(u,v)\in \SE$ to indicate the edge between nodes $u\in\SV$ and $v\in\SV$.  
The graph $\SG=(\SV,\SE)$ grows over time; and we let $\ET$ be the edge arriving at time $t\in\{1,2,...\}$ and $t_{uv}$ be the arrival time of edge $(u,v)$
(i.e., $\ET=(u,v)\Leftrightarrow t_{uv}=t$).
Then, we denote $\SG$ at time $t$ by $\SGT=(\SVT,\SET)$, which consists of the nodes and edges arriving at time $t$ or earlier.
%where, if $\ET=(u,v)$, then $\SVT=\SVTM\cup\{u,v\}$ and $\SET=\SETM\cup\{(u,v)\}$.
Let unordered triple $\triple$ be the triangle (i.e., cycle of length $3$) with edges $(u,v)$, $(v,w)$, and $(w,u)$. 
We use $\tmin$:=$\min\{t_{uv},t_{vw},t_{wu}\}$, $\tmed$:=median$\{t_{uv},t_{vw},t_{wu}\}$, and $\tmax$:=$\max\{t_{uv},t_{vw},t_{wu}\}$ to indicate the arrival times of the first, second, and last edges, resp., in $\triple$.
We denote the set of triangles in $\SGT$ by $\STT$ and the set of triangles with node $u$ by $\STT_{u}\subset\STT$.
We call $\STT$ \textit{global triangles} and $\STT_{u}$ \textit{local triangles} of node $u$.

%\vspace{-1mm}
%\subsection{Density Measure.}
%\label{sec:prelim:density}
%Definition~\ref{defn:density} gives the density measure used in this work.
%That is, the density of a subtensor is defined as the sum of its entries divided by the number of the slices composing it.
%We let $\densityopt$ be the density of the densest subtensor in $\TT$.
%
%\begin{definition}[Density of a subtensor \cite{shin2016mzoom}]\label{defn:density} Let $\asimplesubtensor$ be a subtensor of an $N$-way tensor $\TT$. Then, the density of $\asimplesubtensor$, denoted by $\density{\asimplesubtensor}$, is defined as $$\density{\asimplesubtensor}=\frac{\asum{\asimplesubtensor}}{|S|}$$
%\end{definition}
%
%This measure is chosen because:
%(a) it was successfully used for the purpose of anomaly and fraud detection \cite{shin2016mzoom},
%(b) this measure satisfies axioms that a reasonable ``anomalousness" measure should meet (see Section~A of the supplementary document \cite{supple}),
%and (c) our method based on this density measure outperforms existing methods based on different density measures in Section~\ref{sec:exp:effective:rating}.

\subsection{Problem Definition}
\label{sec:prelim:problem}

In this work, we consider the problem of counting the global and local triangles in a graph stream assuming the following realistic conditions:
\begin{enumerate*}
	\item[C1] \textbf{No Knowledge:} no information about the input stream (e.g., the node count, the edge count, etc) is available. %The length of the input stream can even be infinite.
	\item[C2] \textbf{Real Dynamic:} in the input stream, new edges arrive in the order by which they are created.
	\item[C3] \textbf{Limited Memory Budget:}  we store at most $k$ edges in memory. 
	\item[C4] \textbf{Single Pass:} edges are processed one by one in their order of arrival. Past edges cannot be accessed unless they are stored in memory (within the budget stated in C3).
\end{enumerate*}

Based on these conditions, we define the problem of global and local triangle counting in a real dynamic graph stream in Problem~\ref{defn:problem:merged}. 

\begin{problem}[Global and Local Triangle Counting in a Real Dynamic Graph Stream]\label{defn:problem:merged} \ 
	\begin{enumerate}
		\item \textbf{Given:} a real dynamic graph stream $\{e^{(1)},e^{(2)},...\}$ and a memory budget $k$,
		\item \textbf{Find:} sampling and triangle counting algorithms,
		\item \textbf{to Minimize:} the estimation error of global triangle count $|\STT|$ and local triangle counts $\{|\STT_{u}|\}_{u\in\SVT}$ for each time $t\in\{1,2,...\}$.
	\end{enumerate}
\end{problem}

%\begin{problem}[Global Triangle Counting Problem]\label{defn:problem:global} \ 
%	\begin{enumerate}
%		\item \textbf{Given:} a dynamic graph stream $\{e^{(1)},e^{(2)},...\}$ and a memory budget $k$,
%		\item \textbf{Find:} sampling and triangle counting algorithms,
%		\item \textbf{to Minimize:} the estimation error of the global triangle count $|\STT|$, $\forall t\in\{1,2,...\}$.
%	\end{enumerate}
%\end{problem}
%
%\begin{problem}[Local Triangle Counting Problem]\label{defn:problem:local}\
%	\begin{enumerate}
%		\item \textbf{Given:} a dynamic graph stream $\{e^{(1)},e^{(2)},...\}$ and a memory budget $k$,
%		\item \textbf{Find:} sampling and triangle counting algorithms,
%		\item \textbf{to Minimize:} the estimation error of the local triangle count $|\STT_{u}|$, $\forall u\in \SVT$ and $\forall t\in\{1,2,...\}$.
%	\end{enumerate}	
%\end{problem}

Instead of minimizing a specific measure of estimation error, we follow a general approach of reducing both bias and variance, which is robust to many measures of estimation error.%, as we show in the experiment section.

%Identifying the exact densest subtensor is 
%computationally expensive even for a static tensor. 
%For example, it takes $O(|\aset|^{6})$ even when $\TT$ is a binary matrix (i.e., $N$=$2$) \cite{goldberg1984finding}.
%Thus, we focus on designing an approximation algorithm that maintains a dense subtensor with a provable approximation bound, significantly faster than repeatedly finding a dense subtensor from scratch.
%
%The second problem (Definition~\ref{defn:problem:time}) is to detect suddenly emerging dense subtensors in a tensor stream.
%For a tensor $\TT$ whose values increase over time, let $\TT_{\Delta T}$ be a tensor where the value of each entry is the increment in the corresponding entry of $\TT$ in the last $\Delta T$ time units.
%Our aim is to spot dense subtensors appearing in $\TT_{\Delta T}$, which also keeps changing.
%
%\begin{definition}\textsc{(Detecting Sudden Dense Subtensors in a Tensor Stream)}\label{defn:problem:time}
%	\textbf{(1) Given:} a sequence of increments in a tensor $\TT$ with slices indices $\aset$ (i.e., a tensor stream) and a time window $\Delta T$,
%	\textbf{(2) maintain:} a subtensor $\TT_{\Delta T}(S)$ where $S\subset \aset$, \textbf{(3) to maximize:} $\density{\TT_{\Delta T}(S)}$.
%\end{definition}
%\subsection{Previous Approaches.}

\section{Empirical Pattern: Temporal Locality}
\label{sec:pattern}
%As discussed in Section~\ref{sec:related}, existing streaming algorithms for triangle counting use  uniform random sampling without considering temporal dependencies in edges.
%However, do all edges have ``equal" probability of forming triangles with future edges?

In this section, we discuss {\em temporal locality} (i.e., the tendency that future edges are more likely to form triangles with recent edges than with older edges) in real graph streams.
To show the temporal locality, we investigate the distribution of {\it closing intervals} and {\it total intervals}, defined below.
%Figure~\ref{fig:interval} shows examples of closing and total intervals.
\begin{definition}[Closing Interval]\label{defn:interval:closing}
	The \textbf{closing interval} of a triangle is defined as the time interval between the arrivals of the second edge and the last edge. That is,
	\vspace{-1.5mm}
	\small
	\begin{equation*}
	closing\_interval(\triple):= \tmax-\tmed.
	\end{equation*}\normalfont
\end{definition}
\begin{definition}[Total Interval]\label{defn:interval:total}
	The \textbf{total interval} of a triangle is defined as the time interval between the arrivals of the first edge and the last edge. That is,
	\vspace{-1.5mm}
	\small
	\begin{equation*}
	total\_interval(\triple):= \tmax-\tmin.
	\end{equation*}\normalfont
\end{definition}
Figure~\ref{fig:pattern} shows the distributions of the closing and total intervals in real dynamic graph streams (see Section~D of \cite{supple} for the descriptions of them) and those in random graph streams obtained by randomly shuffling the orders of the edges in the corresponding real streams.
In every dataset, both intervals tend to be much shorter in the real stream than in the random one.
That is, future edges do not form triangles with all previous edges with equal probability but they are more likely to form triangles with recent edges than with older edges.

Then, why does the temporal locality exist? It is related to {\em transitivity} \cite{wasserman1994social}, i.e., the tendency that people with common friends become friends. When an edge $(u,v)$ arrives, we can expect that edges connecting $u$ and other neighbors of $v$ or connecting $v$ and other neighbors of $u$ will arrive soon.
These future edges form triangles with the edge $(u,v)$.
%In other words, each new edge is likely to be the result of a recent edge that forms triangles with the new edge. 
The temporal locality is also related to new nodes.
For example, in citation networks, when a new node arrives (i.e., a paper is published), many edges incident to the node (i.e., citations of the paper), which are likely to form triangles with each other, are created almost instantly.
%and thus their orders in the graph stream are close.
%Since these close edges are incident to the same node (i.e., they are citations from the same paper), they are likely to form triangles with each other.
% which is required for edges forming triangles together, 
%an edge among them is more likely to form triangles with the rest of them than with a random edge in the stream.
Likewise, in social media, new users make many connections within a short time
by importing their friends from other social media or their address books.

\begin{figure}[t]
	\vspace{-3mm}
	\centering
	\includegraphics[width=0.5\linewidth]{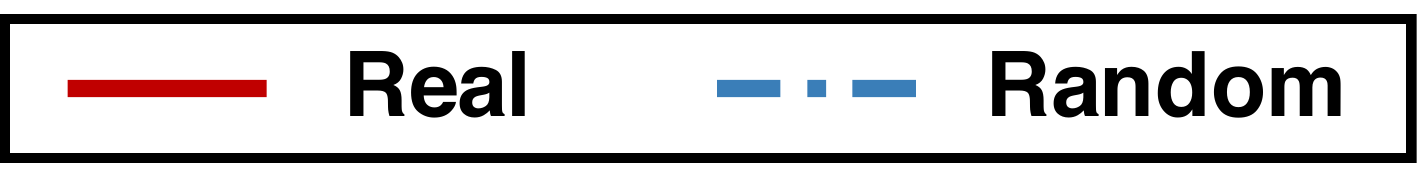} \\
	\vspace{-1.5mm}
	\subfigure[\cintlong Distribution (ArXiv)]{
		\includegraphics[width=0.29\linewidth]{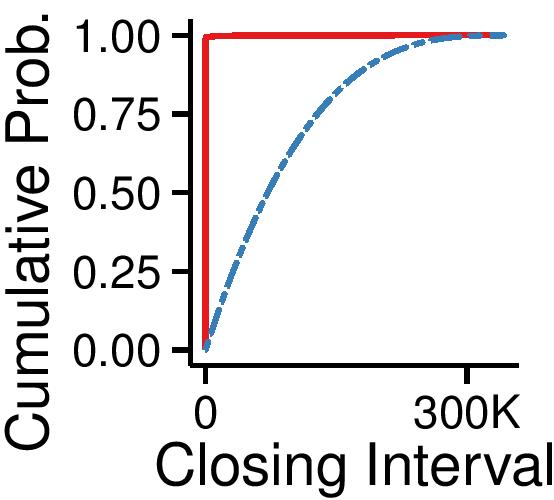}
	}
	\subfigure[\cintlong Distribution (Email)]{
		\includegraphics[width=0.29\linewidth]{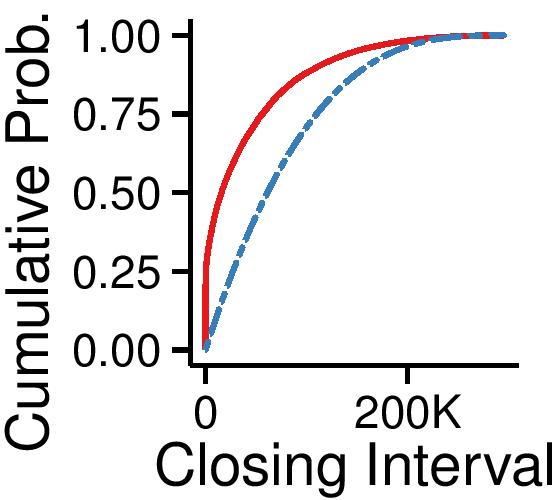}
	}
	\subfigure[\cintlong Distribution (Facebook)]{
		\includegraphics[width=0.29\linewidth]{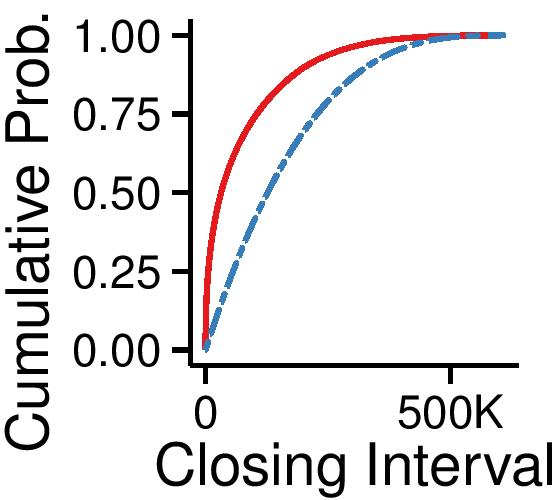}
	}\\
	\vspace{-2mm}
	\subfigure[\fintlong Distribution (ArXiv)]{
		\includegraphics[width=0.29\linewidth]{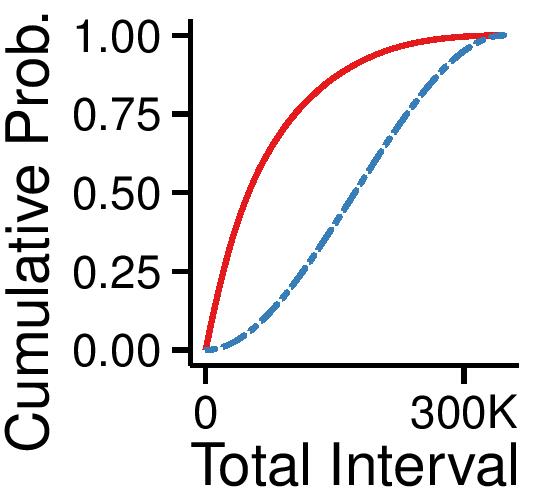}
	}
	\subfigure[\fintlong Distribution (Email)]{
		\includegraphics[width=0.29\linewidth]{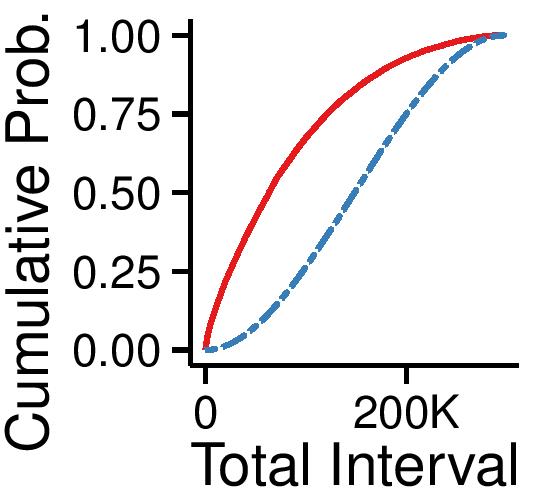}
	}
	\subfigure[\fintlong Distribution (Facebook)]{
		\includegraphics[width=0.29\linewidth]{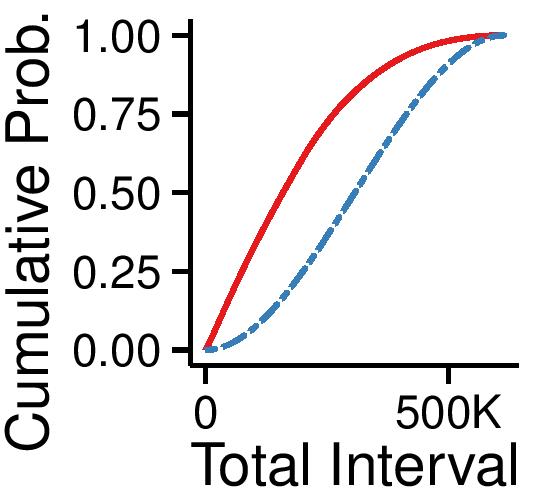}
	}\\
	\vspace{-2mm}
	\caption{\label{fig:pattern}
		{\bf Temporal locality in triangle formation.} Closing and total intervals tend to be shorter in real graph streams than in random ones, i.e., future edges are more likely to form triangles with recent edges than with older edges.
	}
\end{figure}

\section{Proposed Method: Waiting-Room Sampling}
\label{sec:method}
In this section, we propose {\em Waiting-Room Sampling} (\method), a single-pass streaming algorithm that exploits the temporal locality, presented in the previous section, for accurate global and local triangle counting. 
We first discuss the intuition behind \method. Then, we explain the details of \method.
Lastly, we give theoretical analyses of its accuracy. %, and of its time and space complexity in Section~\ref{sec:method:analysis:complexity}.

\subsection{Intuition behind \method}
\label{sec:method:overview}

%Since we assume that at most $k$ edges can be stored in memory, we cannot store all the edges in memory if there are more than $k$ edges in a stream.
%Then, which edges should we store to minimize estimation error of triangle counts?

To minimize the estimation error of triangle counts, we minimize both the bias and variance of estimates.
Reducing the variance is related to finding more triangles because, intuitively speaking, knowing more triangles is helpful to accurately estimate their count.
This relation is more formally analyzed in Section~\ref{sec:method:analysis}.
Thus, the following two goals should be considered when we decide which edges to store in memory:
\bit
	\item {\bf Goal 1.} unbiased estimates of global and local triangle counts should be computed from the stored edges.
	\item {\bf Goal 2.} when a new edge arrives, it should form many triangles with the stored edges. 
\eit

Uniform random sampling, such as reservoir sampling, achieves Goal~1 but fails to achieve Goal~2 ignoring the temporal locality, explained in Section~\ref{sec:pattern}.
Storing the latest edges, while discarding the older ones, can be helpful to achieve Goal~2, as suggested by the temporal locality. However, simply discarding old edges makes unbiased estimation non-trivial. %Then, how can we achieve both goals?

To achieve both goals, \method combines the two policies above.
Specifically, it divides the memory space into the {\em waiting room} and the {\em reservoir}. The most recent edges are always stored in the waiting room, while the remaining edges are uniformly sampled in the reservoir using standard reservoir sampling. %, as described in Figure~\ref{fig:method}.
The waiting room enables us to achieve Goal~2 since it exploits the temporal locality by storing the latest edges, which future edges are more likely to form triangles with.
%, are stored in the waiting room.
%On the other hand, Goal~1 is achieved by the reservoir.
On the other hand, the reservoir enables us to achieve Goal~1, as explained in detail in the following sections.

%Due to the reservoir, when the last edge of a triangle arrives, the probability that the other two edges of the triangle are in memory is strictly positive and is easily computed.
%From this probability, we can obtain unbiased estimates, as explained in detail in the following sections.

\subsection{Detailed Algorithm}
\label{sec:method:algorithm}

We first explain the sampling policy of \method. Then, we explain how to estimate the triangle counts from sampled edges. The pseudo code of \method is given in Algorithm~\ref{alg:method}.

\subsubsection{Sampling Policy (Lines  \ref{alg:method:sample:start}-\ref{alg:method:sample:end} of Algorithm~\ref{alg:method})}
\label{sec:method:algorithm:sampling}
Let $\SSS$ be the given memory space, where at most $k$ edges are stored.
Let $\ET=(u,v)$ be the edge arriving at time $t\in\{1,2,...\}$. 
The sampling method in \method depends on $t$ as follows (see Section~A of \cite{supple} for a pictorial description):

{\bf (Case 1).} If $t\leq k$, add $\ET$ to $\SSS$, which is not full yet.

{\bf (Case 2).} If $t=k+1$, since $\SSS$ is full, divide $\SSS$ into the waiting room $\SW$ and the reservoir $\SR$ so that the latest $k\alpha$ edges (i.e., $\{e^{(k-k\alpha+1)},...,e^{(k)}\}$) are in $\SW$ and the remaining $k(1-\alpha)$ edges (i.e., $\{e^{(1)},...,e^{(k(1-\alpha))}\}$) are in $\SR$.
The constant $\alpha$ is the relative size of the waiting room.
For simplicity, we assume $k\alpha$ and $k(1-\alpha)$ are integers.
Then, go to (Case 3).

{\bf (Case 3).} If $t\geq k+1$, $e^{(t- k\alpha)}$,  which is the oldest edge in $\SW$, is replaced with $\ET$ (i.e., $\SW$ is a queue with the ``first in first out'' mechanism).
Then, with probability $\PT$, where 
\small\begin{equation} 
	\PT:={k(1-\alpha)}/{(t-k\alpha)}, \label{eq:pt} 
\end{equation}\normalsize
$e^{(t-k\alpha)}$ replaces a randomly chosen edge in $\SR$. Otherwise, $e^{(t-k\alpha)}$ is discarded.
That is, standard reservoir sampling \cite{vitter1985random} is used in $\SR$, which ensures that each edge in $\{e^{(1)},...,e^{(t-k\alpha)}\}$ is stored in $\SR$ with equal probability $\PT$.

In summary, when $\ET$ arrives (or after $e^{(t-1)}$ is processed), if $t\leq k+1$, then each edge in $\{e^{(1)},...,e^{(t-1)}\}$ is stored in $\SSS$ with probability $1$.
If $t>k+1$, then each edge in $\{e^{(t-k\alpha)},...,e^{(t-1)}\}$ is stored in $\SSS$ (specifically in $\SW$) with probability $1$, while each edge in $\{e^{(1)},...,e^{(t-k\alpha-1)}\}$ is stored in $\SSS$ (specifically in $\SR$) with probability $p^{(t-1)}$.

%If $t=k$, then $\SSS$ becomes full once $\ET$ is added to $\SSS$.
%In this case, we divide $\SSS$ into the waiting room $\SW$ and the reservoir $\SR$ so that the latest $k\alpha$ are in $\SW$ and the remaining edges are stored in $\SR$.
%
%(a) \textbf{Step 1:} The first $k$ edges are stored in $S$ the order they arrive. The red edge indicates the first edge (i.e., $e_{1}$) and the yellow edge indicates the $k$th edge (i.e., $e_{k}$). 
%(b) \textbf{Step 2:} We divide $S$ into $S_{1}$, where $\alpha k$ latest edges are stored, and $S_{2}$, where the remaining $(1-\alpha)k$ edges are stored.
%(c) \textbf{Step 3:} When the blue edge, which represents $e_{k+1}$, arrives, it is pushed into $S_{1}$ and the green edge, which is oldest in $S_{1}$ is popped. The popped edge is either replaced with a random edge in $S_{2}$ (with prob. $p$) or discarded (with prob. $1-p$). Step 3 is repeated for the remaining edges arriving later.

\begin{algorithm}[t]
	\small
	\vspace{-1mm}
	\caption{Waiting Room Sampling (\method)}\label{alg:method}
	\begin{algorithmic}[1] 
		\Require (1) graph stream: $\{e^{(1)},e^{(2)},...\}$, (2) memory budget: $k$,
		\Statex {\scriptsize \ \ \ \ } (3) relative size of the waiting room: $\alpha$
		\Ensure (1) estimated global triangle count: $c$,
		\Statex	{\scriptsize \ \ \ \ \ \ \ } (2) estimated local triangle counts: $c_{u}$ for each node $u$
		
		%\State $c\leftarrow 0$
		\For{each new edge $e^{(t)}=(u,v)$}
		
		\For{each node $w$ in $\SN_{u}\cap \SN_{v}$ \label{alg:method:count:start}} 
		\State initialize $c$, $c_{u}$, $c_{v}$, and $c_{w}$ to $0$ if they have not been set
		\State increase $c$, $c_{u}$, $c_{v}$, and $c_{w}$ by $1/\tprob$ %\hfill \bluecolor{$\vartriangleright$ Eq. \eqref{eq:tprob}}
		\EndFor \label{alg:method:count:end}
		
		\If{$t \leq  k$} add $e^{(t)}$ to $\SSS$ \label{alg:method:sample:start} \hfill \bluecolor{$\vartriangleright$ (Case 1)} 
		\Else
		\If{$t=k+1$} 
		\vspace{-3.8mm} \Statex \hfill \bluecolor{$\vartriangleright$ (Case 2)}
		\State divide $\SSS$ into $\SW$ and $\SR$ as explained in Section~\ref{sec:method:algorithm:sampling}
		\EndIf 
		\State remove $e^{(t- k\alpha)}$ from $\SW$ and add $e^{(t)}$ to $\SW$ \hfill \bluecolor{$\vartriangleright$ (Case 3)}
		\If{a random number in \textsf{Bernoulli}($\PT$) is $1$ } 
		%\vspace{-3.8mm} \Statex \hfill \bluecolor{$\vartriangleright$ Eq. \eqref{eq:pt}}
		\State replace a randomly chosen edge in $\SR$ with $e^{(t- k\alpha)}$ 
		\EndIf
		\EndIf \label{alg:method:sample:end}
		\EndFor
		%		
		%\Procedure{UpdateCounters}{}($(u,v)$, $t$):
		%
		%
		%\EndProcedure
	\end{algorithmic}
\end{algorithm}

\subsubsection{Estimating Triangle Counts (Lines  \ref{alg:method:count:start}-\ref{alg:method:count:end} of Algorithm~\ref{alg:method})}
Let $\SGH=(\SVH,\SEH)$ be the sampled graph composed of the edges in $\SSS$ ($\SW$ or $\SR$ if $\SSS$ is divided), and let $\SN_{u}$ be the set of neighbors of node $u\in \SVH$ in $\SGH$.
We use $c$ and $c_{u}$ to denote the estimates of the global triangle count and the local triangle count of node $u$, respectively, in the stream so far.
That is, if we let $\CT$ and $\CT_{u}$ be $c$ and $c_{u}$ after processing $\ET$, they are the estimates of $|\STT|$ and $|\STT_{u}|$, respectively.

When each edge $\ET=(u,v)$ arrives, \method first finds the triangles composed of $(u,v)$ and two edges in $\SGH$.
The set of such triangles is $\{\triple:w\in\SN_{u}\cap\SN_{v} \} $.
For each triangle $\triple$, \method increases $c$, $c_{u}$, $c_{v}$, and $c_{w}$ by $1/\tprob$, where $\tprob$ is the probability that \method discovers $\triple$.
Then, the expected increase of the counters by each triangle $\triple$ becomes $1$, which makes the counters unbiased estimates, as shown in Theorem~\ref{thm:unbias} in the following section.

The only remaining task is to compute $\tprob$, the probability that \method discovers triangle $\triple$.
To this end, we divide the types of triangles depending on the arrival times of their edges, as in Definition~\ref{defn:type}.
Recall that $\tgen$ indicates the arrival time of the edge arriving $i$-th among the edges in $\triple$. 

\begin{definition}[Types of Triangles]\label{defn:type}
	Given the maximum number of samples $k$ and the relative size of the waiting room $\alpha$, the type of each triangle $\triple$ is defined as:
	{\small
	$$\type := \begin{cases}
	1 & \textnormal{if}\ \tmax \leq k + 1\\
	2 & \textnormal{else if}\ \tmax-\tmin \leq k\alpha \\
%	  & \textnormal{(or equivalently if $\tmed\in \SW$ and $\tmin\in \SW$)} \\
	3 & \textnormal{else if}\ \tmax-\tmed \leq k\alpha \\
%	 & \textnormal{(or equivalently if $\tmed\in \SW$ and $\tmin\in \SR$)} \\
	4 & \textnormal{otherwise}. \\
	\end{cases}
	$$}
\end{definition}
That is, a triangle has Type~1 if all its edges arrive early, Type~2 if its total interval is short, Type~3 if its closing interval is short, and Type~4 otherwise.
The probability that each triangle is discovered by \method (i.e., considered in line~ \ref{alg:method:count:start} of Algorithm~\ref{alg:method}) depends on its type, as formalized in Lemma~\ref{lemma:probability}.

\begin{lemma}[Triangle Discovering Probability.]\label{lemma:probability}
		Given the maximum number of samples $k$ and the relative size of the waiting room $\alpha$, the probability $\tprob$ that \method discovers each triangle $\triple$ is		
	{\small 
	\begin{equation}
		\tprob = \begin{cases} 
		1 & \textnormal{if}\ \type \leq 2\\
		\frac{k(1-\alpha)}{\tmax-1-k\alpha} & \textnormal{if}\ \type=3 \\
		\frac{k(1-\alpha)}{\tmax-1-k\alpha}\times \frac{k(1-\alpha)-1}{\tmax-2-k\alpha} & \textnormal{if}\ \type=4 \label{eq:tprob} \\
		\end{cases}
	\end{equation}
	}
\end{lemma}
\begin{proof}
	See Section~B.A of \cite{supple}.
\end{proof}

Notice that no additional space is required to store the arrival times of sampled edges.
This is because the type of each triangle $\triple$ and its discovering probability $\tprob$ can be computed from current time $t$ and whether each edge is stored in $\SW$ or $\SR$ at time $t$, as explained in the proof of Lemma~\ref{lemma:probability}.

\subsection{Accuracy Analyses}
\label{sec:method:analysis}

We analyze the bias and variance of the estimates provided by \method.
To this end, we define $\tcount$ as the increase in $\CT$ by triangle $\triple$.
By lines~\ref{alg:method:count:start}-\ref{alg:method:count:end} of Algorithm~\ref{alg:method},  $\tcount$ is $1/\tprob$ with its discovering probability $\tprob$, and $0$ with probability $1-\tprob$.
Based on this concept, the unbiasedness of the estimates given by \method is shown in Theorem~\ref{thm:unbias}.

%Note the following facts about $\tprob$:
%\begin{align}
%	\mathbb{E}[\tcount] & = 1/\tprob\times\tprob+0\times(1-\tprob)=1 \\
%	\mathrm{Var}[\tcount] & = \mathbb{E}[\tcount^{2}] - (\mathbb{E}[\tcount])^2 \\
%	& =   1/\tprob^{2}\times\tprob+ - 1^{2}=   1/\tprob -1.\\
%\end{align}

\begin{theorem}[`Any time' unbiasedness of \method.]\label{thm:unbias} If $k(1-\alpha) \geq 2$,
	\method gives unbiased estimates of the global and local triangle counts at any time.
	That is, if we let $\CT$ and $\CT_{u}$ be $c$ and $c_{u}$ after processing $\ET$, respectively, the followings hold:
	{\small
	\begin{align}
	& \mathbb{E}[\CT]=|\STT|,\ \forall t\in\{1,2,...\}, \label{eq:unbias:global} \\
	& \mathbb{E}[\CT_{u}]=|\STT_{u}|,\ \forall u\in\SVT,\ \forall t\in\{1,2,...\}. \label{eq:unbias:local}
	\end{align}}
\end{theorem}

\begin{proof}
See Section~B.B of \cite{supple}.
\end{proof}

%Figure~\ref{fig:crown:analysis} in Section~\ref{sec:intro} shows the distribution of $1,000$ estimates of the global triangle count obtained by \method in the Arxiv dataset, described in Section~\ref{sec:experiments}. The estimates are distributed so that their mean is close to the true global triangle count, as expected from Theorem~\ref{thm:unbias}.

%Notice that, in order to fully analyze the covariance terms in Eq.~\eqref{eq:var:global:exact} and Eq.~\eqref{eq:var:local:exact}, we have to consider (a) all possible relative arriving orders of the edges consisting a pair of triangles and (b) all possible types of the two triangles in a pair.
%The combination of (a) and (b) results in hundreds of cases.
%Instead of analyzing hundreds of cases, to give a simple intuition, our variance analysis focuses on 
In our variance analysis, to give a simple intuition, we focuses on
\small
\begin{align} \tilde{\mathrm{{Var}}}[\CT] & =\sum\nolimits_{\triple\in\STT}\mathrm{Var}[\tcount], \label{eq:var:global} \\
\tilde{\mathrm{{Var}}}[\CT_{u}] & =\sum\nolimits_{\triple\in\STT_{u}}\mathrm{Var}[\tcount], \label{eq:var:local}
\end{align}
\normalsize
which are the variances when the dependencies in $\{\tcount\}_{\triple\in \STT}$ are ignored.
Specifically, we show how the temporal locality, explained in Section~\ref{sec:pattern}, is related to reducing $\tilde{\mathrm{{Var}}}[\CT]$ and $\tilde{\mathrm{{Var}}}[\CT_{u}]$.

From $\mathrm{Var}[\tcount]=\mathbb{E}[\tcount^{2}]-\left(\mathbb{E}[\tcount]\right)^{2}$,
we have	$\mathrm{Var}[\tcount]=(1/\tprob)-1$.
From Eq.~\eqref{eq:tprob}, if $k(1-\alpha) \geq 2$,
\small
\begin{multline}
\mathrm{Var}[\tcount] = \\
\begin{cases}
	0 & \textnormal{if}\ \type \leq  2\\
	\frac{\tmax-1-k\alpha}{k(1-\alpha)} -1 & \textnormal{if}\ \type=3 \\
	\frac{\tmax-1-k\alpha}{k(1-\alpha)}\times \frac{\tmax-2-k\alpha}{k(1-\alpha)-1} -1 & \textnormal{if}\ \type=4.
\end{cases}
\label{eq:var:individual} 
\end{multline} 
\normalsize Compared to \triest \cite{de2016tri}, where $\mathrm{Var}[\tcount]=0$ if $\type=1$ and $\mathrm{Var}[\tcount]=\frac{\tmax-1}{k}\times \frac{\tmax-2}{k-1}-1$ otherwise, \method reduces the variance regarding the triangles of Type~2 or 3, as formalized in Lemma~\ref{lemma:variance:type}, while \method increases the variance regarding the triangles of Type~4.
\begin{lemma}[Comparison of Variances] \label{lemma:variance:type}
	For each triangle $\triple$, $\mathrm{Var}[\tcount]$ is smaller in \method than in \triest \cite{de2016tri}, i.e., 
	\vspace{-1.5mm}
	\small
		\begin{equation}
		\mathrm{Var}[\tcount] < \frac{\tmax-1}{k}\times \frac{\tmax-2}{k-1}-1 \label{eq:var:comparison}
		\end{equation}
	\normalsize
	if {\bf any} of the following conditions are satisfied:
	\small
	\bit
		\item $\type = 2$
		\item $\type = 3$ and $\tmax > 1+\frac{\alpha}{1-\alpha}k$
		\item $\type = 3$ and $\alpha < 0.5$.
	\eit
	\normalfont
\end{lemma}
\begin{proof}
	See Section~B.C of \cite{supple}.
\end{proof}
Therefore, the superiority of \method in terms of small $\tilde{\mathrm{{Var}}}[\CT]$ and $\tilde{\mathrm{{Var}}}[\CT_{u}]$ depends on the distribution of the types of triangles in real graph streams.
In the experiment section, we show that the triangles of Type 2 or 3 are abundant enough in real graph streams, as suggested by the temporal locality, so that \method is more accurate than \triest.

\section{Experiments}
\label{sec:experiments}
We designed experiments to answer the following questions:
\bit
	\item \textbf{Q1. Accuracy}: How accurately does \method estimate global and local triangle counts?% Is \method more accurate than state-of-the-art streaming algorithms? 
	\item \textbf{Q2. Illustration of Theorems}: Does \method give unbiased estimates with variances smaller than its competitors'?
	\item \textbf{Q3. Scalability}: How does \method scale with the number of edges in input streams?
	%\item \textbf{Q4. Effects of $\alpha$}: How does $\alpha$, the relative size of the waiting room, affect the accuracy of \method? What is the optimal value of $\alpha$?
	%\item \textbf{Q5. Discovery}: What are the discoveries on real-world graph streams?
\eit

\subsection{Experimental Settings}

\noindent{\bf Machine:} We ran all experiments on a PC with a 3.60GHz Intel i7-4790 CPU and 32GB memory.

\noindent{\bf Data:} The real graph streams used in our experiments are summarized in Table~\ref{tab:data:real}. See Section~D of \cite{supple} for the descriptions of them.
In all the streams, edges were streamed in the order by which they are created.%, as assumed in Section~\ref{sec:prelim:problem}.

\begin{table}[t]
	\vspace{-2mm}
	\centering
	\caption{\label{tab:data:real} Summary of real-world graph streams.}
	\begin{tabular}{c|c|c|c}
		\toprule
		{\bf Name} & {\bf \# Nodes} & {\bf \# Edges} & {\bf Summary} \\
		\midrule
		%\textbf{HepTh} \cite{gehrke2003overview} & $27,769$ & $136,190$ & Citation \\
		\textbf{ArXiv} & $30,565$ & $346,849$ & Citation network \\
		%\textbf{Epinion} \cite{massa2005controversial} & $56,499$ & $224,322$ & Trust\\
		\textbf{Facebook} & $61,096$ & $614,797$ & Friendship network \\
		\textbf{Email} & $86,978$ & $297,456$ & Email network \\
		%\textbf{Flickr} \cite{mislove2008growth} & $1,286,641$ & $11,236,021$ & Citation \\
		\textbf{Youtube} & $3,181,831$ & $7,505,218$ & Friendship network \\ 
		\textbf{Patent} & $3,774,768$ & $16,518,947$ & Citation network \\
		\bottomrule
	\end{tabular}
\end{table}

\begin{figure*}
	\vspace{-3mm}
	\centering
	\begin{minipage}{.71\textwidth}
		\centering
		\includegraphics[width=0.65\linewidth]{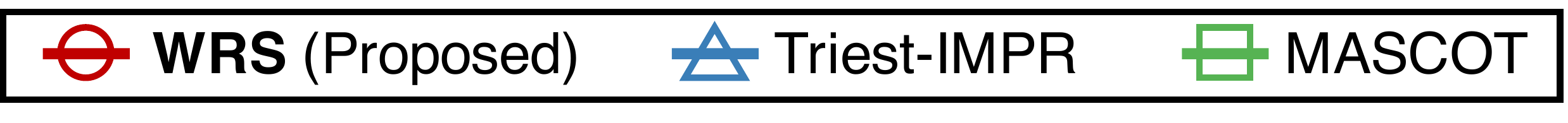} \\
		\vspace{-2mm}
		%	\subfigure[Local Error (ArXiv)]{
		%		\includegraphics[width=0.29\linewidth]{FIG/err_hepph}
		%	}
		\subfigure[Patent (Local)]{
			\includegraphics[width=0.21\linewidth]{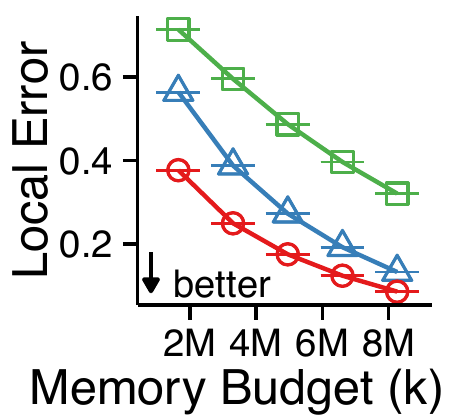}
		}
		\subfigure[Email (Local)]{
			\includegraphics[width=0.21\linewidth]{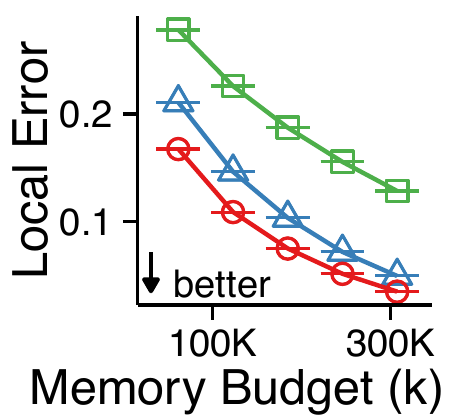}
		}
		\subfigure[Facebook (Local)]{
			\includegraphics[width=0.21\linewidth]{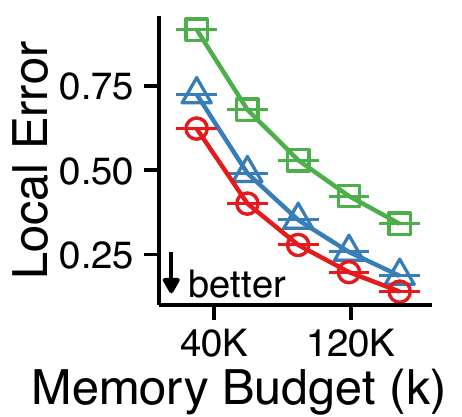}
		}
		\subfigure[Youtube (Local)]{
			\includegraphics[width=0.21\linewidth]{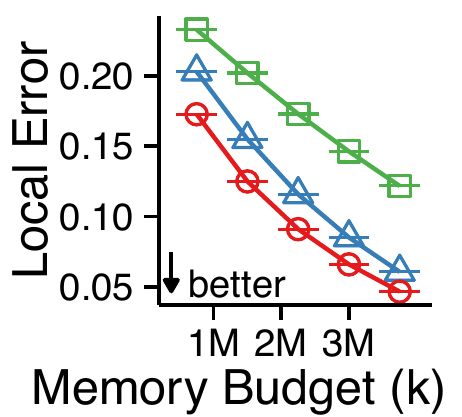}
		} \\ \vspace{-3mm}
		\subfigure[Patent (Global)]{
			\includegraphics[width=0.21\linewidth]{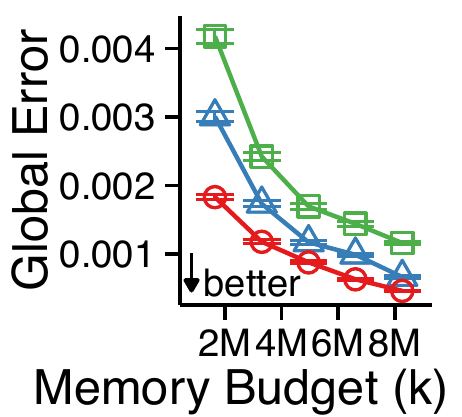}
		}
		\subfigure[Email (Global)]{
			\includegraphics[width=0.21\linewidth]{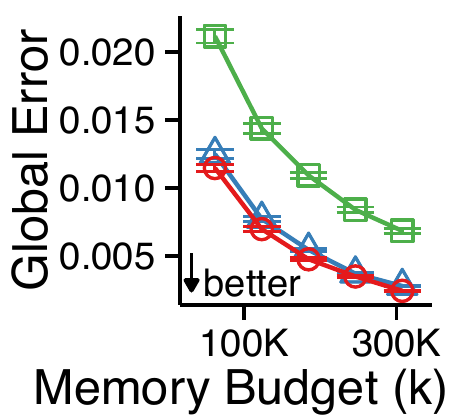}
		}
		\subfigure[Facebook (Global)]{
			\includegraphics[width=0.21\linewidth]{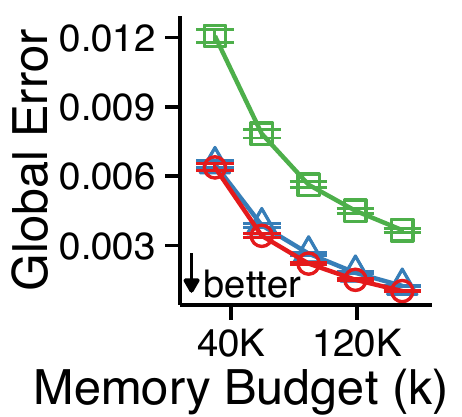}
		}
		\subfigure[Youtube (Global)]{
			\includegraphics[width=0.21\linewidth]{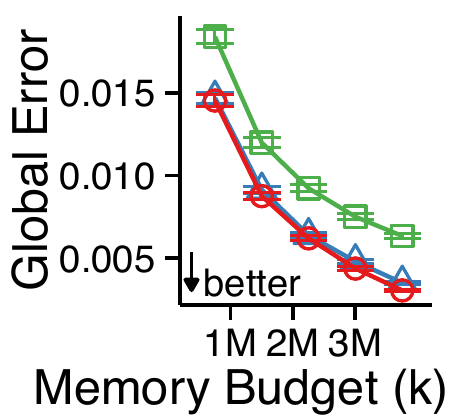}
		} \\
		\vspace{-2mm}
		\caption{\label{fig:accuracy}
			\underline{\smash{\method is accurate.}} M: million, K: thousand. In all the datasets, \method is most accurate in global and local triangle counting regardless of memory budget $k$. %The relative size of waiting room (i.e., $\alpha$) is fixed to $0.1$. %The reported values are the average result of $1000$ runs, and the error bars indicate estimated standard error.
		}
	\end{minipage}
	\hfill
	\begin{minipage}{.25\textwidth}
	\includegraphics[width=1\linewidth]{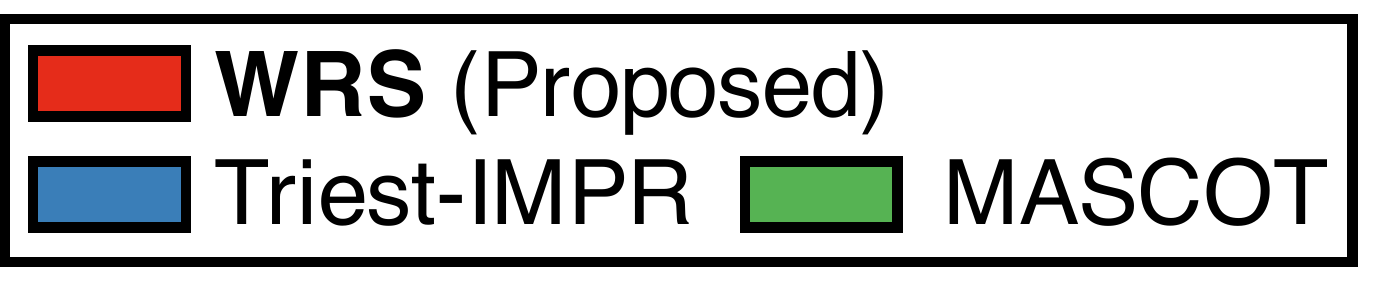} \\
	\includegraphics[width=1\linewidth]{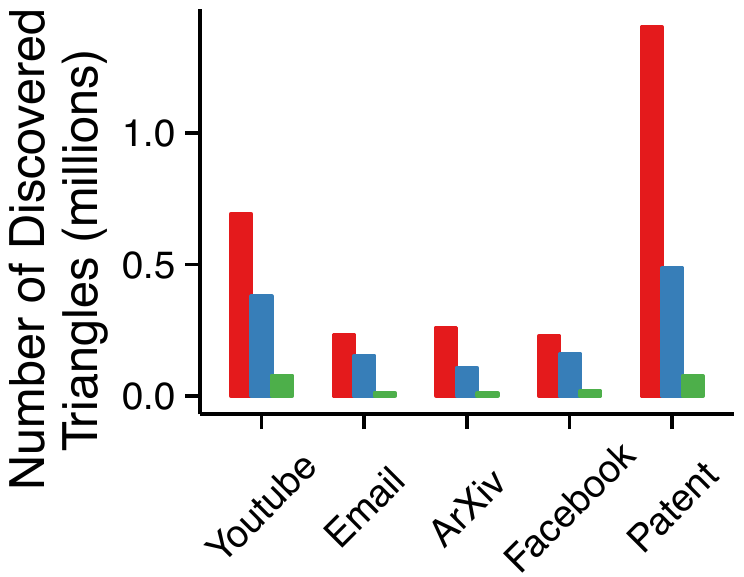} \\
	\vspace{-5.5mm}
	\caption{\label{fig:count}
		\underline{\smash{The sampling scheme}} \underline{\smash{of \method is effective.}} \method discovers up to $2.9\times$ more triangles than the second best method, in the same streams. $k$ is set to $10\%$ of the number of the edges in each dataset.% and $\alpha$ is $0.1$. 
	}	
%		\vspace{-2mm}
%		\centering
%		\subfigure[{\small Elapsed Time}]{
%			\includegraphics[width=0.46\linewidth]{FIG/scalable_machine}
%		}
%		\subfigure[{\small Speed Up}]{
%			\includegraphics[width=0.46\linewidth]{FIG/scalable_machine_speedup}
%		}
%		%\subfigure[Scalability w.r.t $k$]{
%		%	\includegraphics[width=0.185\linewidth]{FIG/scalable_block}
%		%}
%		\caption{\label{fig:scalability:machine}
%			{\bf $\method$ scales out.}
%			\method was speeded up 8$\times$ with 10 machines, and 20$\times$ with 40 machines.
%		}
	\end{minipage}
\end{figure*}

\noindent{\bf Implementations:} 
We compared \method to \triest \cite{de2016tri} and \mascot \cite{lim2015mascot}, which are single-pass streaming algorithms estimating both global and local triangle counts with a limited memory budget.
We implemented all the methods in Java, and in all of them, we stored sampled edges in the adjacency list format.
In \method, the relative size $\alpha$ of the waiting room was set to $0.1$ unless otherwise stated (see Section E.B of \cite{supple} for the effect of $\alpha$ on the accuracy).

\noindent{\bf Evaluation measures:} 
To measure the accuracy of global and local triangle counting, we used the following metrics:
\bit
	\item {\bf Global Error}:
	Let $x$ be the number of the global triangles at the end of the input stream and $\hat{x}$ be an estimated value of $x$. Then, the global error is $|x-\hat{x}|/(x+1)$. %where $1$ is added to the denominator for cases when $x=0$.
	%Lower global relative error indicates higher accuracy of global triangle counting.
	\item {\bf Local Error} \cite{lim2015mascot}:
	Let $x_{u}$ be the number of the local triangles of each node $u\in \SV$ at the end of the input stream and $\hat{x}_{u}$ be its estimate. Then, the local error is
	{\small
		\begin{equation*}
		\frac{1}{|\SV|}\sum\nolimits_{u\in \SV} \frac{|x_{u}-\hat{x}_{u}|}{x_{u}+1}.
		\end{equation*}
	}
	%where $1$ is added to the denominator for cases when $x_{u}=0$.
	%Lower local relative error indicates higher accuracy of local triangle counting.
\eit

%
%
%\begin{figure*}[t]
%	\vspace{-3mm}
%	\centering
%	\begin{tabular}{cl}
%		\begin{minipage}{.47\textwidth}
%			\centering \includegraphics[width=0.8\linewidth]{FIG/space} 
%			\includegraphics[width=0.465\linewidth]{FIG/space}
%			\includegraphics[width=0.465\linewidth]{FIG/space}
%		\end{minipage} 
%		& \begin{minipage}{.3\textwidth}
%			\small
%				\begin{tabular}{c|c|c}
%					\toprule
%					& Recall @ Top-10& Recall @ Top-10\\
%					& in Yelp & in Android \\
%					\midrule
%					\methodA & {\bf 0.9} & {\bf 0.7} \\
%					Simple traffic & 0.0 & 0.0 \\
%					Others 	\cite{hooi2016fraudar,jiang2015general, maruhashi2011multiaspectforensics,shin2016mzoom} & 0.0 & 0.0 \\
%					\bottomrule
%				\end{tabular}
%		\end{minipage}  \vspace{-1.5mm} \\
%		(a) \methodA  in Yelp (left) and Android (right) & (b) Comparison with other anomaly detection methods
%	\end{tabular}
%	%\vspace{-3mm}
%	\caption{\label{fig:review}
%		{\bf \methodA accurately detects small-scale short-period attacks} injected in review datasets.
%		However, these attacks are not distinct in traffic changes (i.e., the number of ratings in each time unit) and overlooked by existing methods.
%	}
%\end{figure*}

\subsection{Q1. Accuracy} \label{sec:exp:accuracy}

Figures~\ref{fig:crown}(b) and \ref{fig:accuracy} show the accuracies of the considered methods in the real graph streams with different memory budgets $k$.
Each evaluation metric was computed $1,000$ times for each method, and the average was reported with an error bar indicating the estimated standard error.
In all the datasets, \method was most accurate in global and local triangle counting, regardless of memory budgets.
The accuracy gain was especially high in the ArXiv and Patent datasets, which showed the strongest temporal locality.
In the ArXiv dataset, for example, \method gave up to {\em 47\% smaller local error} and {\em 40\% smaller global error} than the second best method.

\method was more accurate since its estimates were based on more triangles. Due to its effective sampling scheme, \method discovered up to $2.9\times$ more triangles than its competitors while processing the same streams, as shown in Figure~\ref{fig:count}.

\subsection{Q2. Illustration of Theorems} \label{sec:exp:illustration}

We ran experiments illustrating our analyses in Section~\ref{sec:method:analysis}.
Figure~\ref{fig:crown}(c) shows the distributions of $10,000$ estimates of the global triangle count in the ArXiv dataset obtained by each method. We set $k$ to the $10\%$ of the number of edges in the dataset. % and $\alpha$ to $0.1$.
The average of the estimates of \method was close to the true triangle count.
Moreover, the estimates of \method had smaller variance than those of the competitors.
These results are consistent with Theorem~\ref{thm:unbias} and Lemma~\ref{lemma:variance:type}.

\subsection{Q3. Scalability} \label{sec:exp:scalability}
We measured how the running time of \method scales with the number of edges in input streams.
To measure the scalability independently of the speed of input streams, we measured the time taken by \method to process all the edges ignoring the time taken to wait for the arrivals of edges.
Figure~\ref{fig:crown}(a) shows the results in graph streams that we created by sampling different numbers of edges in the ArXiv dataset.
%Specifically, each graph stream with $s$ edges corresponds to the first $s$ edges in ArXiv dataset. 
%Each reported value is the average over $1,000$ runs.
%The same experiment was performed with the Email and Patent datasets, as shown in Figure~\ref{fig:scalability}.
%In all the considered datasets, 
The running time of \method scaled linearly with the number of edges.
That is, the time taken by \method to process each edge was almost constant regardless of the number of edges arriving so far.

\section{Conclusion}
\label{sec:conclusion}
%In this work, we propose \methodS, an incremental algorithm for detecting dense subtensors in tensor streams, and \methodA, an incremental algorithm spotting sudden dense subtensors.
%They have the following advantages:
%\begin{compactitem}[$\ \bullet$]
%\item {\bf Fast and `any time'}: our methods maintain and update a dense subtensor in a tensor stream, which is up to {\bf a million times faster} than batch methods (Figure~\ref{fig:speed}).
%\item {\bf Provably accurate}: subtensors maintained by our methods have provable lower bounds in density (Theorems~\ref{thm:increment:accuracy}, \ref{thm:decrement:accuracy}, \ref{thm:realtime:accuracy}) and high density in practice (Figure~\ref{fig:accuracy}).
%\item {\bf Effective}: \methodA successfully detects anomalies, including small-scale attacks that existing methods overlook in real-world tensors (Figures~\ref{fig:review} and \ref{fig:kowiki}).
%\end{compactitem}

We propose \method, a single-pass streaming algorithm for global and local triangle counting.
\method divides the memory space into the waiting room, where the latest edges are stored, and the reservoir, where the remaining edges are uniformly sampled.
By doing so, \method exploits the temporal locality in real dynamic graph streams, while giving unbiased estimates.
Specifically, \method has the following strengths:
\begin{itemize*}
\item {\bf Fast and `any time'}: \method scales linearly with the number of edges in input graph streams, and it gives estimates at any time while input streams grow (Figure~\ref{fig:crown}(a)).
\item {\bf Effective}: estimation error in \method is up to {\em $47\%$ smaller} than those in the best competitors (Figures~\ref{fig:crown}(b) and \ref{fig:accuracy}).
\item {\bf Theoretically sound}: \method gives unbiased estimates with small variances under the temporal locality (Theorem~\ref{thm:unbias}, Lemma~\ref{lemma:variance:type} and Figure~\ref{fig:crown}(c)).
\end{itemize*}
{\bf Reproducibility}: The code and data
we used in the paper are available at \textit{\url{http://www.cs.cmu.edu/~kijungs/codes/wrs/}}.

% Camera Ready
%\thanks{
\vspace{1mm}
{\footnotesize 
\noindent{\bf Acknowledgments.}
We thank Prof. Christos Faloutsos and Mr. Jisu Kim from Carnegie Mellon University for fruitful discussions.
This material is based upon work
   supported by the National Science Foundation
   under Grant No.
   % IIS-1247489 % with nikos and tom - umn: IIS-1247632  - mainly: Vagelis - tensors
   % IIS-1217559 % with Niki &  Polo - mainly: Danai, maybe Neil/Jay-Yoon
   CNS-1314632 % with Michalis and Tina - mainly: Vagelis, maybe Neil/Alex - anything that has to do with malware, and/or fraud
   and IIS-1408924. % with Leman - review fraud, BP
   % old
   %% and under the auspices of the U.S. Department of Energy 
   %% by University of California Lawrence Livermore National Laboratory 
   %% % for 2008 or earlier
   %% % under contract No.W-7405-ENG-48.
   %% % under contract DE-AC52-07NA27344 (LLNL-CONF-404625),
   %% % subcontracts B579447, B580840.
   %% %
   %% % for 2009-2010
   %% under contract DE-AC52-07NA27344
   % This work was performed 
   % for 2010, for Hanghang, Lei, Leman 
   % and under the auspices of the U.S. Department of Energy by Lawrence
   % Livermore National Laboratory under contract 
   % # No. DE-AC52-07NA27344. 
   % B594252 # for 2011 - this is for Lei
   % DTRA grant with Abdelzaher, Eliassi-Rad, Han - for U Kang, Vagelis
   % Research was sponsored by the Defense Threat Reduction Agency 
   % and was accomplished under contract No. HDTRA1-10-1-0120.
   % for work on the ARL CTA grant with UIUC - Bruno, Prithwish, Xifeng
   % mainly, for Alex, Vagelis, Miguel, Sunhee, Stephan
   Research was sponsored by the Army Research Laboratory 
   and was accomplished under Cooperative Agreement Number W911NF-09-2-0053. 
   Kijung Shin was supported by KFAS Scholarship.
   Any opinions, findings, and conclusions or recommendations expressed in this
   	material are those of the author(s) and do not necessarily reflect the views
   	of the National Science Foundation, or other funding parties.
   	The U.S. Government is authorized to reproduce and 
   	distribute reprints for Government purposes notwithstanding 
   	any copyright notation here on.
   % The views and conclusions contained in this document are those 
   % of the authors and should not be interpreted as representing 
   % the official policies, either expressed or implied, 
   % of the Army Research Laboratory or the U.S. Government. 
   % The U.S. Government is authorized to reproduce and 
   % distribute reprints for Government purposes notwithstanding 
   % any copyright notation here on.
   %%%% for work on ADAMS/PRODIGAL with SAIC/Raytheon - 
   %%%% Polo, Danai, maybe Leman
   %%%% for 2013: Danai and Jay-Yoon
   %%%% anything that has to do with anomaly detection
   %Funding was provided by the U.S. Army Research Office (ARO) and Defense Advanced Research Projects Agency (DARPA) under Contract Number W911NF-11-C-0088.  The content of the information in this document does not necessarily reflect the position or the policy of the Government, and no official endorsement should be inferred.  The U.S. Government is authorized to reproduce and distribute reprints for Government purposes notwithstanding any copyright notation here on.
   % This work is also partially supported by 
   % an IBM Faculty Award and
   % a Google Focused Research Award.

}

\balance
\bibliographystyle{IEEEtran}
\bibliography{BIB/kijung}

%\appendices
%\input{099appendix}

%\input{100etc}
\end{document}